%% file: main.tex
\documentclass[a4paper,fleqn]{cas-dc}

\usepackage[numbers]{natbib}

\usepackage{makecell}
\usepackage{amssymb}
\usepackage{booktabs}
\usepackage{diagbox}
\usepackage{multirow}
\usepackage{framed}
\usepackage{threeparttable}
\usepackage{graphicx}
\usepackage{pifont}
\usepackage{tabularx}
\usepackage{booktabs}
\usepackage{hyperref}
\usepackage{amsmath}
\usepackage{amsthm}
\usepackage{microtype}

\newtheorem{Theorem}{Theorem}

\def\tsc#1{\csdef{#1}{\textsc{\lowercase{#1}}\xspace}}
\tsc{WGM}
\tsc{QE}

\begin{document}
\let\WriteBookmarks\relax
\def\floatpagepagefraction{1}
\def\textpagefraction{.001}

\shorttitle{Enforcing Cryptographic Distributed-VCS Access Control with No Trust on Servers}

\shortauthors{X. Xu et al.}

\title [mode = title]{Enforcing Cryptographic Distributed-VCS Access Control with No Trust on Servers}  

\author[1]{Xin Xu}
\cormark[1]
\ead{xuxin527@mail.tsinghua.edu.cn}
\author[2]{Zhen Yang}
\author[3]{Quanwei Cai}
\author[4]{Jingqiang Lin}
\author[5]{Liangqin Ren}
\author[6]{Bo Chen}
\author[1]{Yongfeng Huang}

\cortext[1]{Corresponding author}

\affiliation[1]{
  organization={Department of Electronic Engineering, Tsinghua University},
  city={Beijing},
  postcode={100084},
  country={China}
}

\affiliation[2]{
  organization={School of Cyberspace Security, Beijing University of Posts and Telecommunications},
  city={Beijing},
  postcode={100876},
  country={China}
}

\affiliation[3]{
  organization={Beijing Zitiao Network Technology Co., Ltd.},
  city={Beijing},
  postcode={100098},
  country={China}
}

\affiliation[4]{
  organization={University of Science and Technology of China},
  city={Hefei},
  postcode={230022},
  country={China}
}

\affiliation[5]{
  organization={Department of Electrical Engineering and Computer Science, University of Kansas},
  city={Lawrence},
  postcode={66045},
  country={USA}
}

\affiliation[6]{
  organization={Department of Computer Science, Michigan Technological University},
  city={Houghton},
  postcode={49931},
  country={USA}
}

\begin{abstract}
Version control systems (VCS), including central VCS (CVCS) and distributed VCS (DVCS), are widely adopted to manage changes to software code and various types of documents. Unlike CVCS, where entities obtain data from a central server, each entity in DVCS stores the entire repository and shares it independently. In VCS, existing access control schemes require the participation of a central server and cannot be deployed in a completely distributed scenario. Additionally, these schemes often fail to enforce fine-grained access control for write permissions, which is crucial for collaborative work in a distributed environment. In this paper, we propose a \underline{d}istributed \underline{v}ersion control system \underline{a}ccess \underline{c}ontrol scheme (named DVAC), which enforces cryptographic access control on distributed user nodes based on attribute-based encryption (ABE) and attribute-based signature (ABS). DVAC is designed to enforce a cryptographic access control protocol for DVCS, which enables file granularity read and write separation access control without the support of a central server.
To ensure the integrity of the core version control functions in DVCS while protecting data security, DVAC incorporates a version control adaptation protocol. Additionally, DVAC leverages Ethereum smart contracts to maintain access control policies, ensuring distributed storage and trusted management of access policies. The architecture of DVAC is designed to seamlessly integrate with existing mature DVCS, such as Git, with minimal modifications. We have implemented a prototype of DVAC and integrated it with Git. A comprehensive performance evaluation was conducted to assess the overhead introduced by DVAC, and it was demonstrated that the overhead is modest.
\end{abstract}


\maketitle

\input{tex/1_introduction}
\input{tex/2_background}
\input{tex/3_overview}
\input{tex/4_design}
\input{tex/5_disac}
\input{tex/6_security_analysis}
\input{tex/7_evaluation}
\input{tex/8_discussion}
\input{tex/9_conclusion}
\input{tex/10_appendix}

\bibliographystyle{cas-model2-names}

\bibliography{references}

\end{document}

%% file: tex/1_introduction.tex
\section{Introduction}
\label{sec:introduction}
Version control systems (VCS) are crucial for managing data changes and enabling project collaboration. They can be standalone or built-in. VCS can be categorized as central (CVCS) or distributed (DVCS). CVCS, like CVS \cite{cvs} and SVN \cite{svn}, rely on a central server for reference versions, with clients synchronizing their modifications. DVCS, such as Git \cite{git.com} and Veracity \cite{veracity}, use a peer-to-peer approach, where each peer maintains its own repository with a complete change history and synchronizes with other peers. DVCS are more efficient and flexible than CVCS, as they eliminate single points of failure. In DVCS, each entity works on its local repository and performs common operations without communication. Entities can create, maintain, or delete local branches independently. They can also choose which modifications to share and store copies pulled from others, providing data backup and avoiding single points of failure.

DVCS introduce novel challenges to data access control. In this decentralized system, peers have the ability to exchange modifications directly, circumventing the need for a central server to enforce access control policies. Moreover, collaboration in DVCS necessitates the segregation of read and write permissions, thereby amplifying the complexity of access control. The intricacies of access control in DVCS can be succinctly encapsulated as follows: 1) \textbf{Inability to depend on a central server}; 2) \textbf{Uniform implementation across all peers}; 3) \textbf{Separate control of read and write permissions}.

Existing access control schemes for cloud storage and CVCS \cite{Ruj2012Privacy,thwin2019blockchain,wang2019secure,li2024revocable,patil2024secure} are not applicable for DVCS because they require the cooperation of a central server to enforce access control. For example, schemes based on proxy re-encryption \cite{thwin2019blockchain,wang2019secure} and attribute-based cryptosystems \cite{Ruj2012Privacy,li2024revocable,patil2024secure} all rely on a server for various tasks such as re-encryption, verification, and signature checking. 
Distributed access control schemes \cite{shafagh2020droplet, yang2024attribute, li2025efficient, zhang2018bads,exceline2024flexible} also do not meet the requirements of DVCS, as they lack support for write permission control, which is essential for collaboration. These schemes primarily focus on preventing unauthorized read access to sensitive data. However, they are not suitable for DVCS. The distributed read and write framework Disac ~\cite{xu2019enforcing} lacks efficient metadata management and experimental verification. Current access control schemes for Git \cite{gitolite,gitcrypt,git-remote-gcrypt,xu2023gringotts} also face challenges in extending to distributed scenarios, lack write access control, rely on a central server, or involve excessive key maintenance overhead. Therefore, none of the existing access control schemes can effectively address the access control difficulties in DVCS.

In this paper, a novel distributed access control scheme named DVAC is introduced, specifically tailored for DVCS. DVAC utilizes attribute-based encryption (ABE), attribute-based signature (ABS), and smart contracts to achieve its objectives. 
The contributions of this work are summarized as follows:

\begin{itemize}
  \item 
  DVAC provides enforced cryptographic access control for a fully distributed VCS. DVAC can consistently control access to files on distributed user nodes without requiring a central server's involvement.
  DVAC provides granular access control for each file, with permissions specified by the data owner to define the data scope and its corresponding access control policies.
  The access control policy is managed in a distributed and trusted manner by an Ethereum-based smart contract, ensuring consistent multi-node control and distributed access control implementation.

  \item DVAC provides a read and write separation access control protocol and a version control adaptation protocol for DVCS, ensuring data security without disrupting the normal functions of DVCS. We implement read-write separation access control based on attribute-based encryption (ABE) and attribute-based signature (ABS). Users whose attributes do not satisfy the read policy cannot decrypt the ciphertext associated with the data, and modifications made by users whose attributes do not satisfy the write policy are rejected and not accepted by other users. DVAC offers an adaptation protocol for the core version control function of DVCS, ensuring cryptographic protection for data in the repository and on the network, without affecting the version control operation in the work-space. As a result, DVAC prevents unauthorized manipulation of data stored by individual peers without undermining the efficient and flexible version control and collaboration provided by DVCS.
 
  \item 
  We implemented the DVAC prototype integrated with Git, offering trusted metadata management combined with the Ethereum blockchain and lightweight key management, regardless of the number of users.
  DVAC's architecture is highly compatible with existing DVCS and requires minimal configuration.
  Each access control policy is uniquely associated with the data (i.e., files), and a smart contract is utilized to ensure that the access control policy remains under the control of the data owner while being accessible to all entities.
  We have successfully implemented a prototype of DVAC and integrated it with Git, demonstrating its practical application.
  The performance evaluation of DVAC indicates that the introduced overhead is within acceptable limits for practical usage scenarios.
\end{itemize}

The rest of the paper is organized as follows.
Section~\ref{sec:background} introduces the background and related works. 
We present the system overview and threat model in Section~\ref{sec:systemoverview}, and detail the specific protocols and algorithms of DVAC in Section~\ref{sec:scheme}.
Then, we describe how to integrate DVAC with Git in Section~\ref{sec:disacingit}.
The security is analyzed in Section~\ref{sec:security},
and the implementation details and evaluation are provided in Section~\ref{sec:performance}.
Finally, Section~\ref{sec:conclusion} concludes this paper.

%% file: tex/2_background.tex
\section{Background and Related works}
\label{sec:background}
DVAC employs ABE and ABS to enforce access control for read and write operations in DVCS. Furthermore, smart contracts are utilized in DVAC to guarantee the availability and integrity of the access control policy. This section provides crucial background information on DVCS, ABE, ABS, and smart contracts, and concludes with a summary of pertinent research in the field.

\subsection{Distributed Version Control Systems}
DVCS improve project collaboration by providing participants with a complete repository and the ability to make local modifications. This flexibility empowers users to revise data locally, choose modifications to publish, and perform changes based on shared versions. In the absence of a central server, the project is redundantly stored across multiple participants, thereby eliminating single points of failure.

Git, a widely used DVCS, stores a snapshot of each version (called \verb+commit+) instead of just the differences between versions. Git comprises four primary components: workspace, stage, local repository, and remote repositories. Users manipulate files in the workspace, stage modifications using \verb+git add+, create new versions with \verb+git commit+, and share modifications by pushing to remote repositories with \verb+git push+. GitHub serves as a popular remote repository. Users employ commands such as \verb+git clone+ or \verb+git fetch+ to obtain repositories, \verb+git+ \verb+branch+ to create or delete branches, \verb+git checkout+ to switch branches or restore files, and \verb+git merge+ to combine versions. These operations are executed using the Git client.

\subsection{ABE and ABS}
ABE ensures that only entities with attributes satisfying a specified policy can successfully decrypt data \cite{sahai2005fuzzy}. Various ABE algorithms have been proposed, which can be classified into two categories: Key-Policy ABE (KP-ABE) and Ciphertext-Policy ABE (CP-ABE).

In KP-ABE \cite{goyal2006attribute}, access control policies are related to users' secret keys.
In CP-ABE \cite{bethencourt2007ciphertext}, the user's secret key is associated with their attributes, and the data owner specifies the access control policy in the ciphertext. The ciphertext can be decrypted by users whose attributes satisfy the specified policy. CP-ABE provides the data owner complete control over the data and is used in DVAC to enforce read access control. In recent years, a variety of CP-ABE algorithms \cite{yin2024attribute,huang2024efficient} adapted to practical scenarios have been proposed.

CP-ABE\footnote{Hereafter, ABE refers to CP-ABE  for simplicity.} consists of  four functions: $ABE.Setup$, $ABE.\\KeyGen$, $ABE.Encrypt$ and $ABE.Decrypt$.
\begin{itemize}
  \item $ABE.Setup(1^{\lambda})$: Takes a security parameter $1^{\lambda}$ as input and outputs a public parameter $PP_r$ and a master key $MK_r$.
  \item $ABE.KeyGen(MK_r, S_i)$: Takes the master key $MK_r$ and an attribute set $S_i$ of user $i$ as input, and outputs a secret key $SK_r$.
  \item $ABE.Encrypt(PP_r, AP_r, m)$: Takes the public parameter $PP_r$, an access policy $AP_r$, and a message $m$ to be encrypted as input, and outputs the ciphertext $C_m$.
  \item $ABE.Decrypt(PP_r, SK_r, C_m)$: Takes the public parameter $PP_r$, a secret key $SK_r$, and a ciphertext $C_m$ as input, and outputs the decrypted result $m'$.
\end{itemize}

ABS \cite{maji2011attribute} ensures that only entities with attributes satisfying the specified access control policy can generate valid signatures. ABS consists of the following four algorithms.

\begin{itemize}
  \item $ABS.Setup(1^{\lambda})$: Takes a security parameter $1^{\lambda}$ and returns a public parameter $PP_w$ and a master key $MK_w$.
  \item $ABS.KeyGen(MK_w, S_i)$: Takes the master key $MK_w$ and an attribute set$S_i$ for user $i$ and returns a secret key $SK_w$.
  \item $ABS.Sign(PP_w, AP_w, SK_w, m)$: Takes the public parameter $PP_w$, an access policy $AP_w$, secret key $SK_w$ and a message $m$, and returns a signature $\sigma$.
  \item $ABS.Verify(PP_w, \sigma, m)$: Takes the public parameter $PP_w$, a signature $\sigma$ and a message $m$. Returns 1 if the signature is valid, otherwise returns 0.
\end{itemize}

\subsection{Smart Contract}
Smart contracts, introduced by Nick Szabo in 1994, are ``computerized transaction protocols that execute the terms of a contract'' \cite{smartcontract}. They aim to fulfill common contractual conditions (e.g., payment terms, liens) without the need for intermediaries  \cite{christidis2016blockchains}. 

Smart contracts are facilitated by blockchain platforms like Ethereum \cite{wood2014ethereum} and Hyperledger Fabric \cite{hyperledger}. As a widely deployed blockchain, Ethereum enables developers to create smart contracts using languages like Solidity, compile them, and deploy them. The smart contracts are stored and executed on the Ethereum virtual machine. To deploy a smart contract, the owner initiates a transaction, creating a contract account. Any Ethereum account interacts with the smart contract through the contract account. When an Ethereum account invokes a function to update data stored on the Ethereum network, it submits a corresponding transaction, which is recorded on the Ethereum network.

\subsection{Related Works}
In this section, we divide the work into three categories: centralized access control schemes, distributed access control schemes, and access control schemes for Git.

\textbf{Centralized access control schemes.} In centralized access control schemes, a central server stores access control-related information (such as access control lists) and participates in the execution of controls. However, these schemes cannot scale to distributed scenarios because of the significant overhead imposed on access control enforcement entities. For example, in the scheme \cite{thwin2019blockchain,wang2019secure} based on proxy re-encryption, the entity needs to perform independent re-encryption for each version and each delegated user. Similarly, in schemes based on selective encryption \cite{Vimercati2007Over,Sabrina2013Enforcing}, the entity needs to dynamically maintain the key derivation graph, which is not suitable for a larger number of cooperators. 
In contrast, DVAC only requires each entity to perform one ABE encryption and ABS signature for each version, regardless of the number of cooperators.

\textbf{Distributed access control schemes.} Distributed access control schemes like ABE, ABS, and smart contracts have been used for fine-grained access control on distributed data. However, these schemes are not suitable for DVCS, as they either require a central server to assist and cannot meet the fully distributed requirements\cite{Ruj2012Privacy,Zhao2011Realizing,li2024revocable,patil2024secure}, or they cannot provide write access control and do not support user collaboration\cite{zhang2018bads,shafagh2020droplet,yang2024attribute, li2025efficient}.

Both schemes \cite{li2024revocable, patil2024secure} propose access controls for electronic health records (EHR). RVWABE-CA \cite{li2024revocable} introduces a weighted attribute-based encryption scheme with revocation and verifiability, while CP-ABSC-MIoT \cite{patil2024secure} offers a privacy-preserving access control scheme based on attribute signcryption. However, both focus solely on read permissions for EHR and omit write permission control. Additionally, both depend on a central Verification Server or Authorization Engine, limiting flexibility in distributed environments. Similarly, the schemes \cite{Ruj2012Privacy,Zhao2011Realizing} use ABE and ABS for data confidentiality and integrity, but also rely on a central server for revision checks.

BaDS \cite{zhang2018bads} uses smart contracts for access control in the Internet of Things (IoT), while Droplet \cite{shafagh2020droplet} develops a decentralized authorization service for fine-grained access control. Exceline et al. \cite{exceline2024flexible} propose an EHR access control mechanism based on smart contracts, but the plaintext data remains unprotected during the smart contract-assisted operations. Yang et al.\cite{yang2024attribute} introduce an attribute-based access control scheme for IoT using blockchain technology. Li et al.\cite{li2025efficient} present an efficient ciphertext-policy weighted attribute-based encryption scheme for multi-user collaborative access in cloud storage. However, these schemes mainly focus on confidentiality, lack support for write access control, and are not applicable to DVCS where read and write permissions are separate. Xu et al. \cite{xu2019enforcing} propose a simple distributed read/write framework, but metadata is protected only by digital signatures, requiring users to maintain the public key certificates of data owners. Additionally, the application and evaluation in real DVCS scenarios are not addressed.

\textbf{Access control schemes for Git.} Several schemes have been proposed for securing Git. Gitolite  \cite{gitolite} provides access control for Git in a central server model. Git-crypt \cite{gitcrypt} encrypts sensitive information in public repositories using Git filters \cite{gitfilter}. Git-remote-gcrypt \cite{git-remote-gcrypt} encrypts entire repositories, distributing decryption keys through GPG. GitHub Enterprise \cite{gitee} provide access control services based on the role-based access control (RBAC) model. Keybase-git \cite{gitkeybase} is built on the Keybase Encrypted File System (KBFS), and uses the ACL model for access control. In 2023, Xu et al. proposed Gringotts \cite{xu2023gringotts}, an end-to-end encrypted VCS with file-level read access control using ABE and branch-level write access control using ECDSA. However, Gringotts relies on a remote server for maintaining an encrypted shadow repository and lacks flexibility in DVCS scenarios. DVAC integrates with Git through Git filters, providing fine-grained protection for each file in the repository. Key management in DVAC is greatly simplified because users only need to maintain their own read and write attribute keys, and data owners only need to set access policies based on attributes, without managing other users' related data.

%% file: tex/3_overview.tex
\section{System Overview}
\label{sec:systemoverview}

For DVCS with distributed data storage, determining how to enforce and maintain consistent access control for distributed nodes is one of the most important challenges. 
The access control scheme should not require the participation of the central server, forcing the protection of file data on each distributed node, and the current access control scheme requires the participation of the server, which is easily bypassed by the node in DVCS.

\begin{table}[pos=htbp]
    \centering
    \caption{Symbol table.}
    \label{tab:symbol_table}
    \begin{tabularx}{\columnwidth}{@{}p{1.5cm}X@{}}
        \toprule
        Symbol & Description \\
        \midrule
        $SK_i^r$ & The read key of user $i$, used for ABE.Dec. \\
        $SK_i^w$ & The write key of user $i$, used for ABS.Sign. \\
        $ID_f$   & The unique identifier of the file. \\
        $ID_o$   & The unique identifier of a data owner. \\
        $AP_f^r$ & Read access policy for file $f$. \\
        $AP_f^w$ & Write access policy for file $f$. \\
        $Meta_f$ & Metadata for file $f$. \\
        $Cipher_f$ & Ciphertext of the file $f$ content. \\
        $Sig_f$   & ABS signature of the file $f$. \\
        List$_u$ & The collection of mappings between $ID_o$s and blockchain accounts. \\
        List$_f$ & The collection of all $ID_f$s created. \\
        List$_{AP}$ & The collection of all metadata (including historical versions) indexed by $ID_f$s. \\
        \bottomrule
    \end{tabularx}
\end{table}

We summarized the following three difficulties:
\begin{enumerate}
 \item Enforce cryptographic access control across all DVCS nodes and separately control read and write permissions;
\item Access control must be fully compatible with DVCS features such as version control;
\item Access control-related data (such as access policies) requires distributed and secure storage and trusted management to provide consistent access control across multiple nodes.
\end{enumerate}

In distributed version control scenarios, DVAC is designed to eliminate the user's reliance on a central server while ensuring accurate fine-grained read/write access control. Based on this premise, our objective is to minimize alterations and impact to existing DVCS and achieve seamless integration and user-friendliness of DVAC without requiring data migration or changes to normal operations. With these design goals in mind, we present a system overview of DVAC in this section.

As shown in Figure~\ref{fig:systemmodel}, DVAC comprises three modules: enforcing cryptographic access control (EAC), version control adaptation (VCA) and trusted metadata management (TMM). DVAC involves the smart contract deployed on the blockchain, distributed user repositories and attribute authorities. Attribute authorities are responsible for user registration, key distribution and maintenance (such as $PP$), and are associated with the initial deployment of DVAC, which is not detailed here. The smart contract is deployed on the blockchain and is responsible for managing and maintaining the metadata version to ensure its  correctness and validity. Users have a complete repository on their local storage, which is a distributed node in the DVCS. Following local modifications, users can opt to push them to other repositories. Users possess varying permissions for different files, categorized as data owners and readers/writers. Data owners have full control over the files they create and can specify and modify access policies, while readers/writers need to choose to read or modify data according to their own permissions.

\begin{figure}[pos=htbp]
\centering
\setlength{\abovecaptionskip}{-0cm}
\includegraphics[scale=0.59]{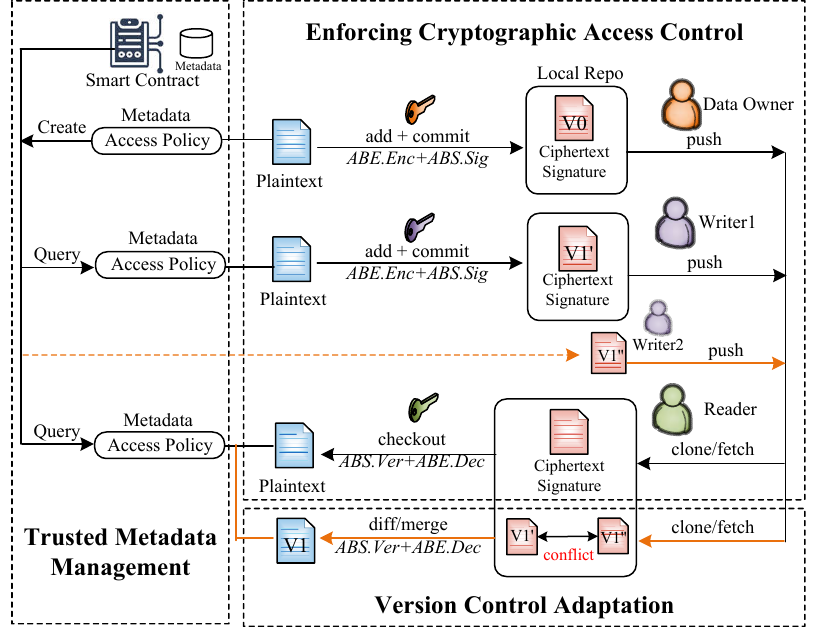}
\caption{The architecture of DVAC.}\label{fig:systemmodel}
\end{figure}

The EAC module is responsible for providing enforcement and consistent access control for files stored in distributed user repositories within DVCS. It also safeguards the historical versions of files stored in the user's local repository, which contain file ciphertexts and signatures.
Ensuring that the original DVCS version control function is not affected when the files in the repository are all ciphertext is a crucial problem to address. 
We designed the VCA module to ensure the accuracy of DVCS version control of encrypted files. When users execute version control commands, appropriate data conversion and adjustment are performed to ensure the accurate execution of commands.
The TMM module is essential for ensuring proper access control. It manages versions of metadata (including access policies) and maintains metadata security.

The attribute authority is responsible for distributing attribute keys to users. As shown in Figure~\ref{fig:systemmodel}, the data owner creates a plaintext file $V0$, encrypts and signs it according to the access policy, and stores it in the local repository. 
In this case of cryptographic protection, the file can be shared with other users, who can then perform access control operations based on access policies and their own attribute keys. The metadata containing the access policy is managed by the smart contract, and users have the capability to request metadata from the smart contract at any time during access control operations. 
Suppose that two modifiers create the $V1$ version of a file at the same time, which is called $V1'$ and $V1''$. This will cause version conflicts and require version control operations such as merging and comparing. DVAC will retrieve files in plaintext via the version control adaptation module for further DVCS operations.

\textbf{Threat Model.} In DVAC, each $AA$ is trusted and securely stores sensitive information, such as $MK_r$, $MK_w$, and user secret keys. $AA$ authenticates users and verifies their attributes before generating their keys. Public repositories hosting the data are honest-but-curious, storing all versions correctly but potentially performing unauthorized read or write operations. Valid users verify received versions using ABS signatures, decrypt data with their secret keys ($SK_r^i$), and generate valid ABS signatures for modified data using their secret keys ($SK_w^i$). Malicious users may engage in unauthorized activities, such as attempting to infer unauthorized data, uploading invalid modifications, or launching DDoS attacks. Collusion among malicious users is also possible to gain additional privileges.

%% file: tex/4_design.tex
\section{DVAC Scheme Design}
\label{sec:scheme}
This section provides a detailed introduction to the specific design and protocol process of DVAC, along with proposed mitigation ideas for the challenging problems raised in the previous section. 

Before delving into the specific protocols of DVAC, we establish the symbolic representation of the basic cryptographic algorithms as follows:
\begin{itemize}
  \item $C\leftarrow Enc(M,k)$: Use the key $k$ to perform symmetric encryption (e.g., AES) on the plaintext $M$ to obtain the ciphertext $C$.
  \item $M\leftarrow Dec(C,k)$: Use the key $k$ to perform symmetric decryption on the ciphertext $C$ to obtain the plaintext $M$.
  \item $s\leftarrow Sig(M,sk)$: Asymmetric cryptographic algorithm (such as RSA). Use the private key $sk$ to sign the message $M$ to obtain the signature $s$.
  \item $0/1 \leftarrow Verify(s, M, pk)$: Use the public key $pk$ to verify the signature $s$, output $1$ means the signature is valid, otherwise it is invalid.
\end{itemize}

\subsection{Initialization Phase}
\label{subsec:init}
During the initialization phase of DVAC, the $ABE.Setup$ and $ABS.Setup$ procedures are initially performed by the attribute authority to generate the associated public parameters ($PP_r, PP_w$) and master keys ($MK_r, MK_w$), which is followed by the user registration process. Each user needs to register with a trusted authority to obtain an authenticated identity ID. The protocol of register is as follows:

In real-world application scenarios, there are typically one or more trusted institutions responsible for verifying the user's actual identity in order for the user to register an account. The trusted institution provides the user with a DVCS account containing a verifiable ID.

\begin{enumerate}
 \item The trusted institution generates a pair of signature keys $(sk_a, pk_a)$;
 \item Authenticate the user $ID$ and generate a pair of signature keys $(sk_{ID}, pk_{ID})$ for the user;
 \item Distribute $(sk_{ID}, cert)$ for the registered user $ID$, where $cert=(Info_{ID}, s) = ((ID, pk_{ID}), s)$, $s=Sig(Info_{ID}, sk_a)$.
\end{enumerate}

In DVAC, users collaborate securely in DVCS with the help of blockchain. As a result, users go through the registration process using the blockchain smart contract. 
\begin{enumerate}
 \item The user $ID$ creates a blockchain account $Acc$;
 \item The user $ID$ initiates a registration request $Req(reg,\\ Acc, ID, cert, s')$ to the smart contract, where $s'=Sig(Acc|cert, sk_{ID})$.
 \item The smart contract verifies $s$, $0/1 \gets Verify(s, \\Info_{ID}, pk_a)$. If $0$ is returned, the error code ``cert error'' is printed.
 \item The smart contract verifies $s'$, $0/1 \leftarrow Verify(s',\\ Acc|cert, pk_{ID})$. If $0$ is returned, the error code ``user error'' is printed.
 \item The smart contract adds $(Acc, ID) \rightarrow List_u $.
\end{enumerate}

In DVCS, the data are stored in the form of files, and DVAC adds extra information for each file to enforce the access control. Firstly, we introduce file-related data structures designed in DVAC.
\begin{itemize}
  \item \textbf{File ID}: The identity of each file ($f$) need to be unique and easily derivable for users. In DVAC, the file ID $ID^f$ is defined as $ID^f = (Init\_CID|R\_Name|\\F\_Path)$, where $Init\_CID$ is the digest of the first commit information (e.g., creator, creating time, included files, etc.) of the repository, and $F\_Path$ is the relative file path. $Init\_CID$ and $R\_Name$ can uniquely refer to a repository, and $F\_Path$ can uniquely specify a file in the repository. 
  \item \textbf{Metadata}: Metadata refers to information about file and access policies. We define the metadata of a file as $Meta_{f}$ is $Meta_{f}=(ID^o, ID^f, AP^f_r, AP^f_w)$, where $ID^o$ and $ID^f$ are the unique identity of the data owner and the file, and $AP^f_r$ and $AP^f_w$ are the read and write access policies of the file, respectively. 
  $Meta_{f}$ is indexed by $ID^f$, bound to file ($f$) one by one, and the rest part $ID^o, ID^f, AP^f_r$ and $AP^f_w$ can only be modified by the data owner $ID^o$. 
  \item \textbf{File}: In DVAC, each file ($f$) is stored in the form of [$Meta_{f}$, $Cipher_{f}$, $Sig_{f}$], where $Cipher_f$ contains the ciphertext of $f$, while $Sig_f$ is the ABS signature of $Cipher_f$. 
\end{itemize}

Here, we establish a unique and easily generated $ID^f$ for each file to ensure that the user can accurately index the metadata using the $ID^f$, thereby avoiding confusion resulting from ID duplication and data owner preemption.

\subsection{Enforcing Cryptographic Access Control (EAC)}
In DVAC, users only need to possess their own pair of read and write attribute keys for file-level access control in various repositories. This reduces the burden of key management for users.  In DVCS environments, files resemble blockchain ledgers and can be stored in a distributed manner in local repositories of each participating user. Users who are unable to decrypt files using the read attribute key do not have read access, and similarly, users who are unable to generate valid signatures using the write attribute key do not have write access. We describe the EAC module in detail below, including EAC protocols and adjustments to basic DVCS file operations.

As shown in Figure~\ref{fig:protocol1.1}, when creating a file, the data owner needs to set the read and write access policies ($AP_r, \\AP_w$) for the file. DVAC will generate the file ID $ID^f$ and synthesize the metadata $(ID^o, ID^f, AP^f_r, AP^f_w)$ to initiate a metadata creation request to the smart contract. The metadata creation process can be synchronized with the process of the data owner writing the contents of the file. Upon completion of both processes, DVAC performs $ABE.Enc$ and $ABS.Sign$ according to the read and write access policy and the attribute key$(SK_r,SK_w)$, and stores the generated $[Cipher, Sig]$ in the local repository of the data owner. The data owner has the option to share the file $[Cipher, Sig]$ with other users.

Then, we show how to enforce the read-write file protocol on a distributed client as shown in Figure~\ref{fig:protocol1.2}. Suppose that the client stores three versions of the file $(V_0, V_1, V_2)$, and the file content is in the form of $[Cipher,Sig]$. To read file $V_2$, DVAC obtains the  $ID^f$ and queries the metadata of the corresponding version from the smart contract. After obtaining the metadata, DVAC first verifies the signature $Sig_2$ according to the write access policy $AP_w$. If the verification is successful, $User_2$ has the write permission and the file version $V_2$ is valid. 

\begin{figure}[pos=htbp]
\centering
\setlength{\abovecaptionskip}{-0cm}
\includegraphics[scale=0.65]{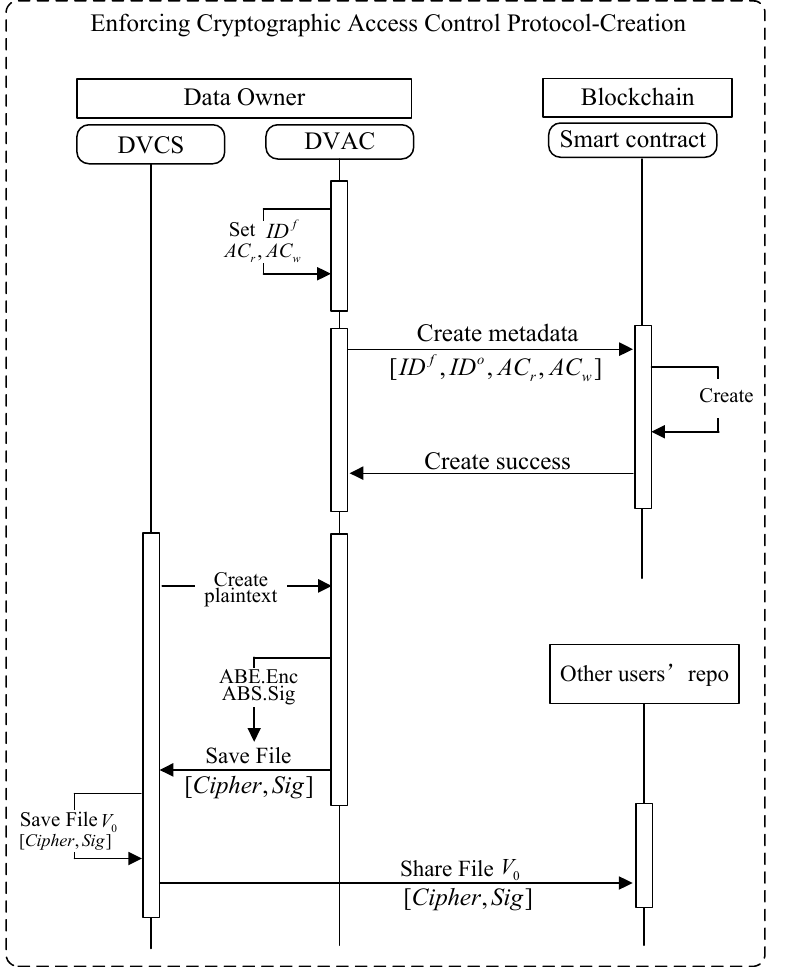}
\caption{Enforcing cryptographic access control protocol with file creation.}\label{fig:protocol1.1}
\end{figure}

Next, DVAC will decrypt $Cipher_2$ according to $AP_r$ and the user's attribute key. If the user's attributes meet the access policy $AP_r$, the decryption is successful, and the plaintext file will be displayed in the client workspace. If the signature is invalid, $User_2$ does not have write permission. DVAC will return an error warning and choose to verify the latest historical version $V_1$ according to user requirements until a valid version is verified. If decryption fails, DVAC returns an error warning indicating no permissions. When a user modifies a file and intends to store it as a new version, DVAC encrypts the file according to the access policy and public parameter, and signs the ciphertext with the user's attribute key $SK_w$, obtaining $[Cipher_3, Sig_3]$. The user's local repository stores the $V_3$ version of the file, and can share it with other users.

The EAC protocol ensures that access control can be distributed on user nodes without the involvement of a central server, enforcing cryptographic protection. The data obtained from other nodes is in the form of $[Cipher,Sig]$. The plaintext data can be read only after being verified and decrypted by the user. Unauthorized modification can only affect the local repository of the modifier and cannot be accepted by other users.

\begin{figure}[pos=htbp]
\centering
\setlength{\abovecaptionskip}{-0cm}
\includegraphics[scale=0.65]{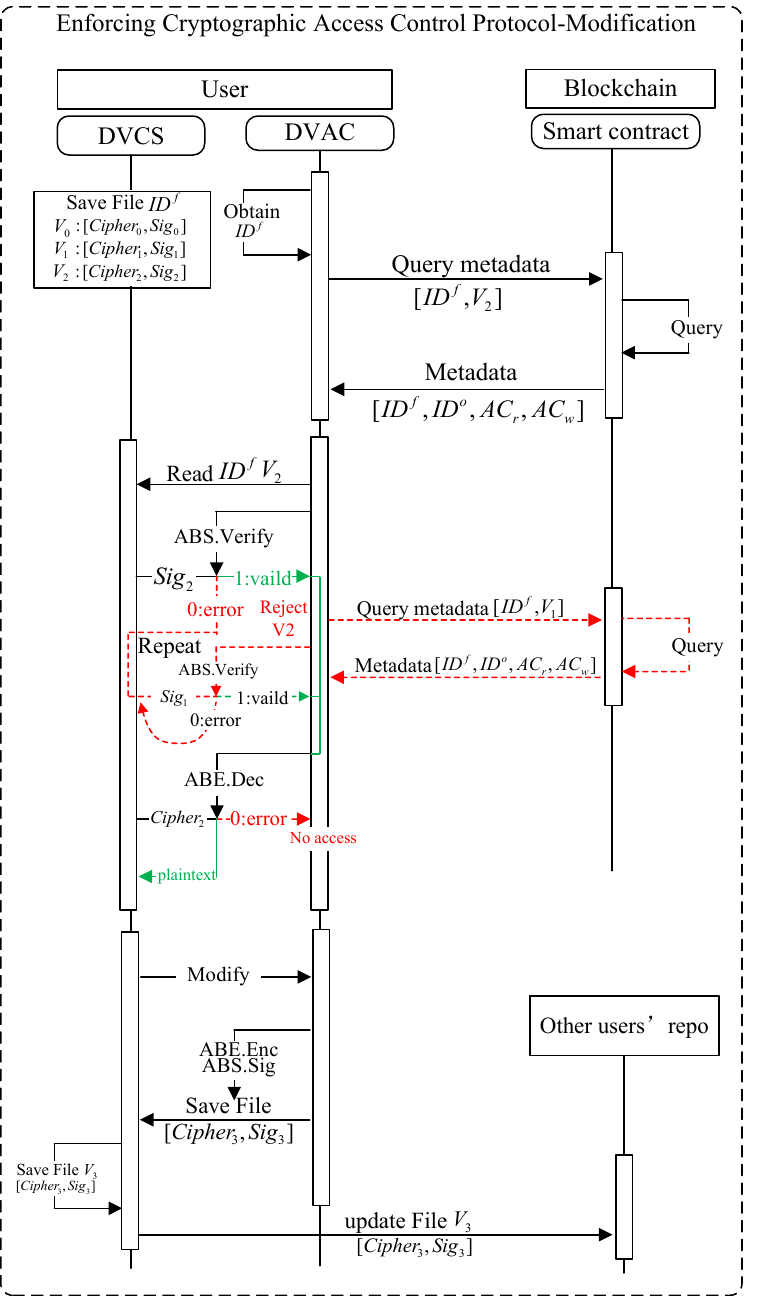}
\caption{Enforcing cryptographic access control protocol with other file operations.\protect\footnotemark}\label{fig:protocol1.2}
\end{figure}
\footnotetext{$V_0$ is created by the data owner. Assume that $V_1$ is the version modified by $User_1$ and $V_2$ is the version modified by $User_2$.}

\subsection{Version Control Adaptation (VCA)}
DVAC ensures that files in the user repository are encrypted. Users need to verify and decrypt $[Cipher,Sig]$ to obtain plaintext data. Each file version is encrypted with a distinct key, resulting in a different ciphertext. The version control function of DVCS, encompassing the maintenance, merging, and comparison of different versions, does not directly address the management of ciphertext. Consequently, we have developed an adaptation protocol for version control.

DVAC needs to release the cryptographic protection of related files when users execute version control commands and display the plaintext result of the command execution in the workspace. Meanwhile, the files in the repository remain encrypted. As shown in Figure~\ref{fig:protocol2}, the user's local repository stores three versions of file $ID^f$. The user initiates the diff/merge command for $V_1$ and $V_2$ in the workspace. Subsequently, DVAC sends a metadata query request to the smart contract with $ID^f$, $V_1$, and $V_2$ as additional information. Upon receiving the metadata corresponding to the two versions from the smart contract, DVAC retrieves $[Cipher_1,Sig_1]$ and $[Cipher_2,Sig_2]$ from the user repository, verifying and decrypting the files individually. In the event of a verification or decryption failure, a warning message is sent to the user workspace. DVAC then forwards the successful plaintext of $V_1$ and $V_2$ to the DVCS for automatic diff/merge operation, and the final result is displayed in the workspace.

\begin{figure}[pos=htbp]
\centering
\setlength{\abovecaptionskip}{-0cm}
\includegraphics[scale=0.60]{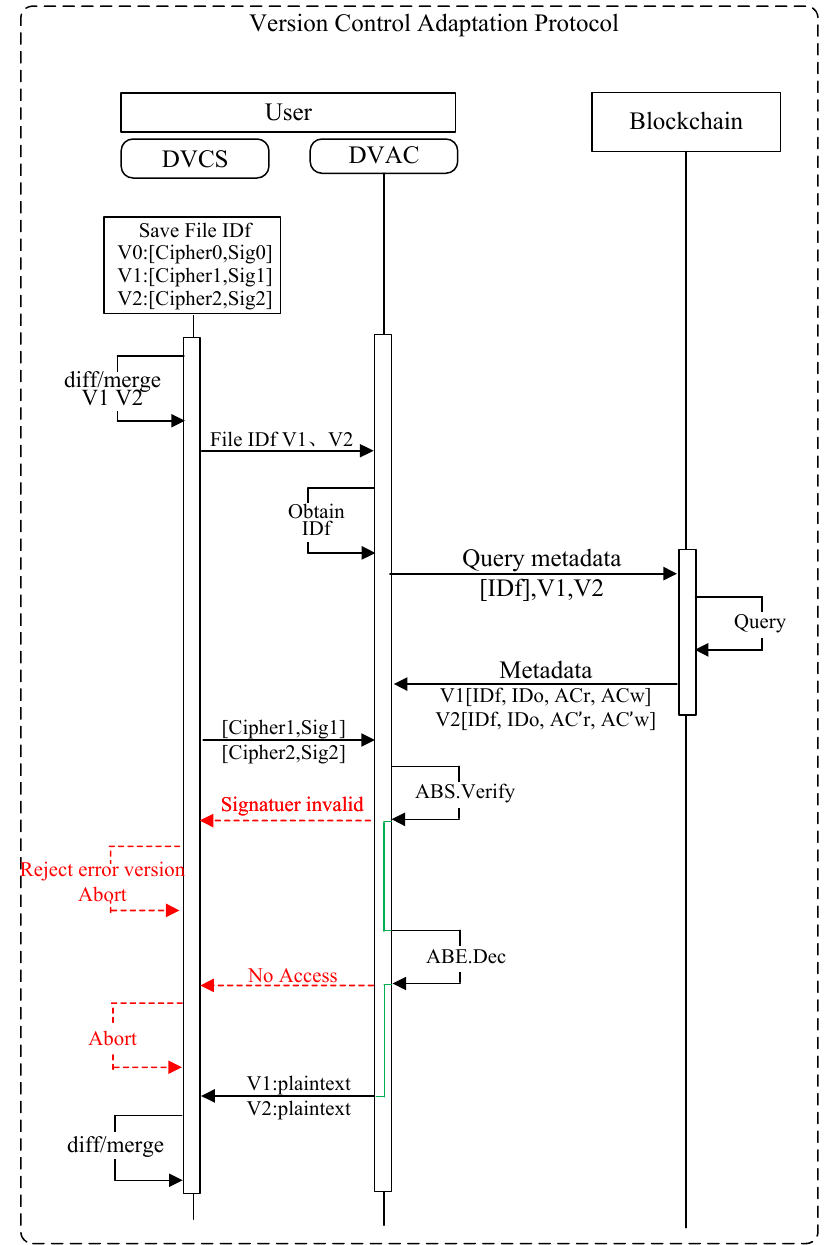}
\caption{Version control adaptation protocol.}\label{fig:protocol2}
\end{figure}

\subsection{Trusted Metadata Management (TMM)}
As described in Section~\ref{subsec:init}, DVAC metadata contains information about files,  as well as read and write access policies. The metadata is mapped to a file using a unique $ID^f$. Different versions of the file may correspond to different versions of the metadata. The security of metadata determines the effectiveness of access control. However, current DVCS lack a proper mechanism to store and maintain DVAC-related metadata for each file version. Once the maintenance of metadata has security vulnerabilities, it will pose additional risks to the security of DVCS data.

In a distributed scenario, if metadata is stored in DVCS in the same form as normal files, users cannot build a global view of the repository and are vulnerable to \textbf{metadata replay attacks}. The attacker may attempt to obtain revoked permissions, by replaying old metadata in which his permissions are not revoked. The victim, who has not received the latest version of the data, may accept this metadata and the attacker's revisions. Due to the inability of users to have a comprehensive overview of all file metadata, obtaining the required correct version of metadata independently is challenging. Additionally, metadata relies on $ID^f$ for file mapping and $ID^o$ for determining the data owner. Attackers may create a large number of files in advance and delete them to \textbf{preempt data owner attacks}. Even if the $ID^f$ is uniquely generated, there is no way to avoid the extreme situation where two data owners create files with the same name and the same path in one repository, leading to the problem of confusing the ownership of the data. The fundamental cause of these attacks is the inability of the data owner and user to build a global view of the repository to check if $ID^f$ has been allocated and obtain the corresponding valid metadata. However, building a global view will introduce the significant performance overhead, which becomes an obstacle in applying DVAC in the practical DVCS.

In order to solve the above problem, we have devised a blockchain-based mechanism for maintaining metadata, ensuring distributed and trustworthy management for DVAC-related metadata.
The characteristics of blockchain, such as addition-only, immutability, and distributed ledger, are highly suitable for fully distributed environments.
By leveraging smart contracts, we have developed a distributed and trusted mechanism for managing metadata.

As shown in Figure~\ref{fig:metadata}, the TMM mechanism encompasses four algorithms: metadata creation, metadata modification, owner information replacement, and metadata querying.

Users are required to submit a request $Req(function, \dots)$ to the smart contract, with the first parameter specifying the action the smart contract should undertake (register, create, modify, replace or query).

The smart contract includes two fundamental functions: 
$Exist(a,List)$ to check if  $a$ exists in $List$, and $Get(a,\\List)$ to obtain an item based on the index of $a$ in $List$. 

The creation $Md\_Create()$ and modification \\$Md\_Modify()$ of metadata necessitate the smart contract to validate the $ID^o$ and the correctness of the account binding  before appending the new metadata to $List_{AP}$, as shown in Figure~\ref{fig:metadata}. The distinction lies in the fact that metadata creation requires the retrieval of $List_f$ to ensure that the $ID^f$ is unused, thus avoiding metadata conflicts and resisting metadata preemption attacks. Owner information replacement $Md\_Replace()$ is divided into $ID^o$ replacement and bound account $Acc^o$ replacement, denoted by a $tag$. Metadata querying $Md\_Query()$ does not mandate authentication of the requester's identity, and the smart contract only indexes and outputs the metadata based on $ID^f$ and version $v$.

\begin{figure}[pos=htbp]
    \begin{framed}
    \raggedright{\textbf{Metadata Management}
    \underline{$Meta_{f} \leftarrow Md\_Create(ID^o,ID^f, AP^f_r, AP^f_w, Acc')$}:\\
 \begin{enumerate}
 \item[1.] Query $ID^o$, $(Acc^o, ID^o) \gets Get(ID^o, List_u)$;
 \item[2.] If $Acc'= Acc^o$, continue, otherwise output "error 1";
 \item[3.] Query $ID^f$, $0/1 \gets Exist(ID^f, List_f)$. If $0$ is returned, output "error 2". If $1$ is returned, $ID^f \rightarrow List_f$.
 \item[4.] $(ID^f,ID^o, AP^f_r, AP^f_w) \rightarrow List_{AP}$, output $Meta_{f}=(ID^f,ID^o, AP^f_r, AP^f_w)$.
\end{enumerate}
    \underline{$Meta_{f} \leftarrow Md\_Modify(ID^o,ID^f, AP^f_r, AP^f_w, v, Acc')$}:
 \begin{enumerate}
 \item[1.] Query $ID^o$, $(Acc^o, ID^o) \gets Get(ID^o, List_u)$;
 \item[2.] If $Acc'= Acc^o$, continue, otherwise output "error 1";
 \item[3.] Obtain $Meta_{f}$ of version $v$,  $Meta_{f}\gets Get(ID^f, List_{AP}, v)$;
 \item[4.] $AP^f_r, AP^f_w \rightarrow Meta_{f}$, output new $Meta_{f}$.
\end{enumerate}

    \underline{$Meta_{f} \leftarrow Md\_Replace(ID^f,ID^o, ID^{o'}/Acc^{o'}, tag)$}:
 \begin{enumerate}
 \item[1.] Query $ID^o$, $(Acc^o, ID^o) \gets Get(ID^o, List_u)$
 \item[2.] If $Acc'= Acc^o$, continue, otherwise output "error1";
 \item[3.] $tag = 1$:
  $(Acc^{o'}, ID^o)\rightarrow List_u$
 \item[4.] $tag=2$:\\
 1)  Obtain $Meta_{ID^f}$ of version $v$,  $Meta_{f}\gets Get(ID^f, List_{AP}, v)$;\\
 2)  $ID^{o'} \rightarrow Meta_{f}$.
\end{enumerate}
 \underline{$ Meta_{f} \leftarrow Md\_Query(ID^f, v):$}
    \begin{enumerate}

 \item[1.] Obtain $Meta_{f}$ of version $v$,  $Meta_{ID^f}\gets Get(ID^f, List_{AP}, v)$;
 \item[2.] Output  $Meta_{f}$.
\end{enumerate}
    }
    \end{framed}
    \caption{Metadata management algorithms}
    \label{fig:metadata}
\end{figure}

%% file: tex/5_disac.tex
\section{Disac Applied in Git}
\label{sec:disacingit}

DVAC can be readily applied in DVCS. In this context, we use the widely adopted DVCS, Git, as an example to illustrate the practical integration process. We introduce the integration of DVAC in Git  and elucidate the specific process for users utilizing Git.

For Git integration, we have implemented a Git filter driver, and associated it with a long-running filter \cite{longrunningfilter} for a specified subset of files by configuring Git attributes for the repository.
The long-running filter, supported since Git 2.11, transmits the file path and content to the driver, and allows a long-running process to operate multiple files instead of only one in the previous version of Git filter.
The registered filter driver includes the processes for two commands: clean and smudge.

\begin{itemize}
  \item The clean command is invoked when a user's modification is added from the workspace to the stage through \verb+git add+. In DVAC, as shown in Figure~\ref{fig:gitdesign}, the process for clean command initially computes the $ID^f$, retrieves the metadata of $ID^f$ by querying the smart contract, performs $ABE.Enc$ and $ABS.Sign$ with $AP_r$ and $AP_w$, and ultimately returns the ciphertext and signature to the Git client for storage in the stage.
  \item The smudge command is invoked when a user pulls the revisions from the remote/local repository or stage to the workspace through the commands such as \verb+git pull+ and  \verb+git checkout+. In DVAC, as shown in Figure~\ref{fig:gitdesign}, the process for the smudge command derives $ID^f$, queries the smart contract for the latest metadata, verifies the ABS signature with $AP_w$,  obtains the plaintext after performing ABE decryption with $AP_r$, and finally returns the plaintext file to the workspace.
\end{itemize}

\begin{figure}[pos=htbp]
  \centering
  \includegraphics[scale=0.6]{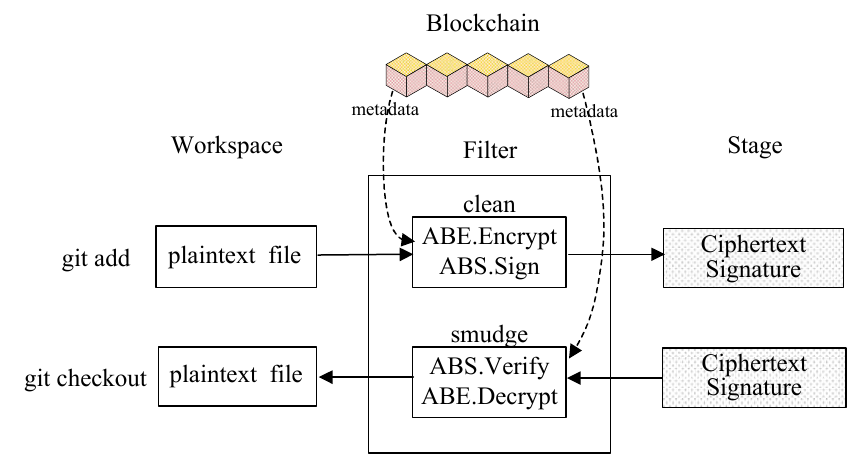}
  \caption{DVAC design integrated with Git.}\label{fig:gitdesign}
\end{figure}

We have developed a web interface based on web3.js \cite{web3js} for the data owner to manage the metadata. This includes the $Md\_Create$ function for establishing  the initial access control policy for a new $ID^f$, the $Md\_Modify$ function for altering the access control policy, the $Md\_Replace$ function for transferring ownership and the account, and the $Md\_Query$ function for retrieving  the correct metadata.
For users, the data exists in plaintext in the workspace, and the VCA protocol ensures that standard Git operations (such as \verb+git merge+ or conflict resolution) remain unaffected.
For the \verb+git rm+ operation, instead of deleting the file, DVAC clears the file content and stores the ABS signature (including the parameters) of the \verb+git rm+ operation in the file, which allows other users to calculate $ID^f$ and check the write permission. The encrypted and signed contents of the file are ultimately stored in the $./objects$ folder of the repository.

%% file: tex/6_security_analysis.tex
\section{Security Analysis}
\label{sec:security}
In this section, we present a concise security analysis of DVAC.

\begin{Theorem}
    DVAC provides file granularity, read permission security for unauthorized read users and write permission security for unauthorized write users.
\end{Theorem}

\begin{proof}
    In a file $f$ [$Meta_{f}$, $Cipher_{f}$, $Sig_{f}$], the attribute set $S_a$ of unauthorized read user $A$ does not meet $AP_r$, and the attribute set $S_b$ of unauthorized write user $B$ does not meet $AP_w$. After obtaining $Cipher_{f}$, user $A$ needs to decrypt the $Cipher_{f}$ to obtain the plaintext. Because $S_a$ does not satisfy $AP_r$ and returns an error when executing $ABE.Dec$, plaintext cannot be obtained by user $A$.  User $B$ reads file $f$ and makes modifications to obtain file $f^{\prime}$. When $ABS.Sign$ is executed based on $AP_w$ in $Meta_{f}$, an error is returned because $S_b$ does not meet $AP_w$, and the correct signature cannot be obtained. However, file $f^{\prime}$ [$Cipher_{f^{\prime}}$, $Sig_{f^{\prime}}$] is generated if user $B$ chooses a forged ($S_b$ satisfied) $AP'_w$ to sign, or a forged signature. Other users obtain the correct $AP_w$ through the smart contract after obtaining the file $f^{\prime}$, and return 0 after executing $ABS.Verify$($Sig_{f^{\prime}}$,$AP_w$). Therefore, all users who execute DVAC correctly will not accept the file $f^{\prime}$, meaning that the modifications made by unauthorized write user $B$ are invalid.
\end{proof}

\begin{Theorem} DVAC is resistant to data owner preemption and impersonation attacks targeting metadata.
\end{Theorem}
\begin{proof} 
  Against the two types of attacks:
  \begin{itemize}
  \item For the data owner preemption attack, the attacker will create a large number of file metadata in advance to overwrite the $ID^f$ that may be used to become the corresponding data owner in advance. In DVAC, the creation of metadata must make a $Req(create)$ to the smart contract. In the process of calling $Md\_Create$, it is necessary to determine whether the created $ID^f$ has been recorded in the $List_f$. Thus, the uniqueness check in $Md\_Create$ ensures that the invoker will select a globally unique and unused $ID^f$, making preemption of $ID^f$ impossible. In addition, metadata creation requires a transaction to be initiated and costs a certain amount of ETH, introducing the significant overhead for the attackers attempting frequent invocations to preempt a larger number of $ID^f$.
  \item For a data owner impersonation attack, an attacker may modify the $ID^o$ in the metadata with $Md\_Replace$, or by directly tampering with the blockchain. For $Md\_Replace$, as long as the user's private key remains secure, the attacker cannot forge a valid signature to impersonate the victim, and therefore the consistency check of the invoker's identifier and stored $ID^o$ in $Md\_Replace$ prevents the attacker from successfully invoking $Md\_Replace$. Direct tampering on the blockchain is prevented by the tamper-resistant feature of the smart contract (with the underlying blockchain).
\end{itemize}

Therefore, DVAC can resist data owner preemption and impersonation attacks targeting metadata.
\end{proof} 

\begin{Theorem} DVAC can protect metadata from tampering and resist metadata replay attacks.
\end{Theorem}
\begin{proof}
    The attackers may attempt to tamper with $AP_r$ or $AP_w$ in the metadata to gain more privileges over the victim's data. To tamper with the metadata, the attacker may choose to modify it by invoking $Md\_Modify$, directly tamper with $AP_r$ and $AP_w$ stored in the blockchain, or replay the $AP_r$ and $AP_w$.
\begin{itemize}
  \item For $Md\_Modify$, DVAC requires the invoker's identifier to match the stored $ID^o$ for $ID^f$, and the attacker cannot forge a valid signature corresponding to the public key for the stored $ID^o$. Consequently, the attacker cannot successfully impersonate the victim to invoke $Md\_Modify$.
  \item Direct modifications to $AP_r$ and $AP_w$ are also infeasible due to the tamper-resistant nature of the blockchain.
  \item Metadata replay attacks cannot succeed either. The attackers cannot perform the read (write) operation once the corresponding permission is revoked, as correct users can obtain a global view of the metadata for $ID^f$ easily through the query interface of the smart contract, so they are therefore able to obtain the correct metadata for each version of the data.
\end{itemize} 
Therefore, DVAC can protect metadata from tampering and resist metadata replay attacks.
\end{proof}

In addition, the anti-collusion ability of DVAC depends on the anti-collusion ability of ABE and ABS algorithms.
The data owner can easily manage the metadata, while the user can obtain any valid metadata through the query interface, without any extra user-side verification.

%% file: tex/7_evaluation.tex
\section{Evaluation}
\label{sec:performance}
In this section, we provide the implementation details and performance evaluation of DVAC.

\begin{table*}[]
  \centering
\caption{Comparison of access control solutions for Git.}\label{tab:compare}
 \begin{threeparttable}
 \setlength{\tabcolsep}{1.25mm}{
\begin{tabularx}{\textwidth}{l l l l l l l l}
\hline

     &  git-remote-gcrypt & Git-secret& git-crypt       & GitHub Enterprise & Keybase-git  &  Gringotts & DVAC in Git
\\ \hline
read permission control\tnote{a}     & R &  R &  F    & R & R    &    F  &  F
\\ 
write permission control\tnote{a} &  \ding{56} &  \ding{56}  & \ding{56}     & B     & R  & B & F
\\ 

key maintenance\tnote{b}  & $O(N_u)$ & $O(N_u)$ & $O(N_u)$    & $O(N_u)$  & $O(N_u)$  & $O(N_u)$ & 2
\\
Distributed scenario\tnote{c} & \ding{56} & \ding{56} & \ding{56}  & \ding{56} & \ding{56}  & \ding{56} & \ding{51}
\\
Version control adaptation & - & - & -  & - & -  & \ding{51} & \ding{51}
\\
Data in repos\tnote{d} & P & P & P  & P & P  & C & C+S
\\\hline
\end{tabularx}}
\begin{tablenotes}
\scriptsize
    \item[a] Read/write permission control. " \ding{56}": not supported; "R": repository level; "B": branch level; "F": file level.
    \item[b] Key maintenance. " $N_u$": the number of users.
    \item[c] Distributed scenario: no centralized server for file storage or access control implementation. " \ding{56}": not supported or excessive encryption overhead; "\ding{51}": support and acceptance of expenses.
    \item[d] Data in repos. "P": plaintext; "C": ciphertext; "C+S": ciphertext+signature.
\end{tablenotes}
\end{threeparttable}
\end{table*}

\subsection{Functional Comparison}
\label{subsec:comparison}
We conducted a functional comparison of DVAC with existing access control solutions for Git, evaluating them based on five aspects: control granularity of read and write permissions, key maintenance overhead, distributed scenario support, and version control adaptation. Our findings indicate that the current Git access control schemes lack the ability to offer fine-grained write access control and struggle to support distributed scenarios, as shown in Table~\ref{tab:compare}.

Both git-remote-gcrypt \cite{git-remote-gcrypt} and Git-secret \cite{gitsecret} solely offer encryption protection at the repository level and do not support write access control. They utilize GnuPG (GPG) \cite{gpg} to maintain the public key of all users and provide distribution of symmetric keys. 
While git-remote-gcrypt uses AES-128 for data encryption, git-crypt \cite{gitcrypt}, similar in nature, employs AES-256 (in CTR mode) for data encryption and can provide file-granular encryption protection.
None of these solutions support write access control, and the public key maintained by GPG is either stored on the public key server or maintained locally by the user. The $O(N_u)$ key maintenance overhead also limits the expansion of distributed scenarios. 

GitHub Enterprise \cite{gitee} provides access control services based on the role-based access control (RBAC) model, offering repository-level read access control and branch-level write access control. However, the implementation of access control requires verification by the GitHub server. Keybase-git \cite{gitkeybase} is built on the Keybase Encrypted File System (KBFS), and uses the ACL model for access control. Users are also required to maintain $N_u$ public key information. While these approaches encrypt data during file transfer, the files remain in plaintext in the user's local repository, thus obviating the need for version control adaptation. Nevertheless, this method cannot guarantee the security of the user's local data, leaving it vulnerable to potential leakage.

Gringotts \cite{xu2023gringotts} adopts ABE to encrypt files in the repository, and uses the elliptic curve digital signature algorithm (ECDSA) to provide branch-level write permission control. While ABE allows users whose attributes meet the access policy to decrypt files, the remote server still needs to verify the signature to enforce write access control and maintain $N_u$ user public keys.
In the distributed scenario, there is no central server for access control implementation and file storage, and it is difficult to directly apply to existing DVCS systems (such as Git). In addition, coarse-grained write access control at the branch level requires users to manually generate branches and merge them, which seriously affects user experience in fully distributed scenarios. In fact, Gringotts is a targeted version control system with access control deployed, and its version control function is closely related to the access control function, so the scheme cannot be easily migrated to other DVCS, such as Git. So the above solutions make it difficult to complete the decentralized deployment and application in the DVCS.

DVAC, integrated with Git, adopts ABE and ABS algorithms to provide file-granularity read and write access control, enabling fine-grained data protection. Through attribute-based cryptographic algorithms, the encryption and signature are calculated according to the attributed access policy. Unlike traditional public key cryptographic algorithms, there is no need to perform encryption for each user and maintain a large number of user public keys. Each user only maintains their own ABE and ABS private keys. Furthermore, DVAC's read and write access control can be fully implemented on the client side, transparent to users, without the need for a central server. It supports distributed scenarios and is well-suited for distributed version control systems such as Git.

\subsection{Performance}
We have evaluated the performance overhead of enforcing read and write access control introduced by DVAC and compared it to existing schemes.

\begin{figure*}[ht]
  \centering
  \includegraphics[scale=0.32]{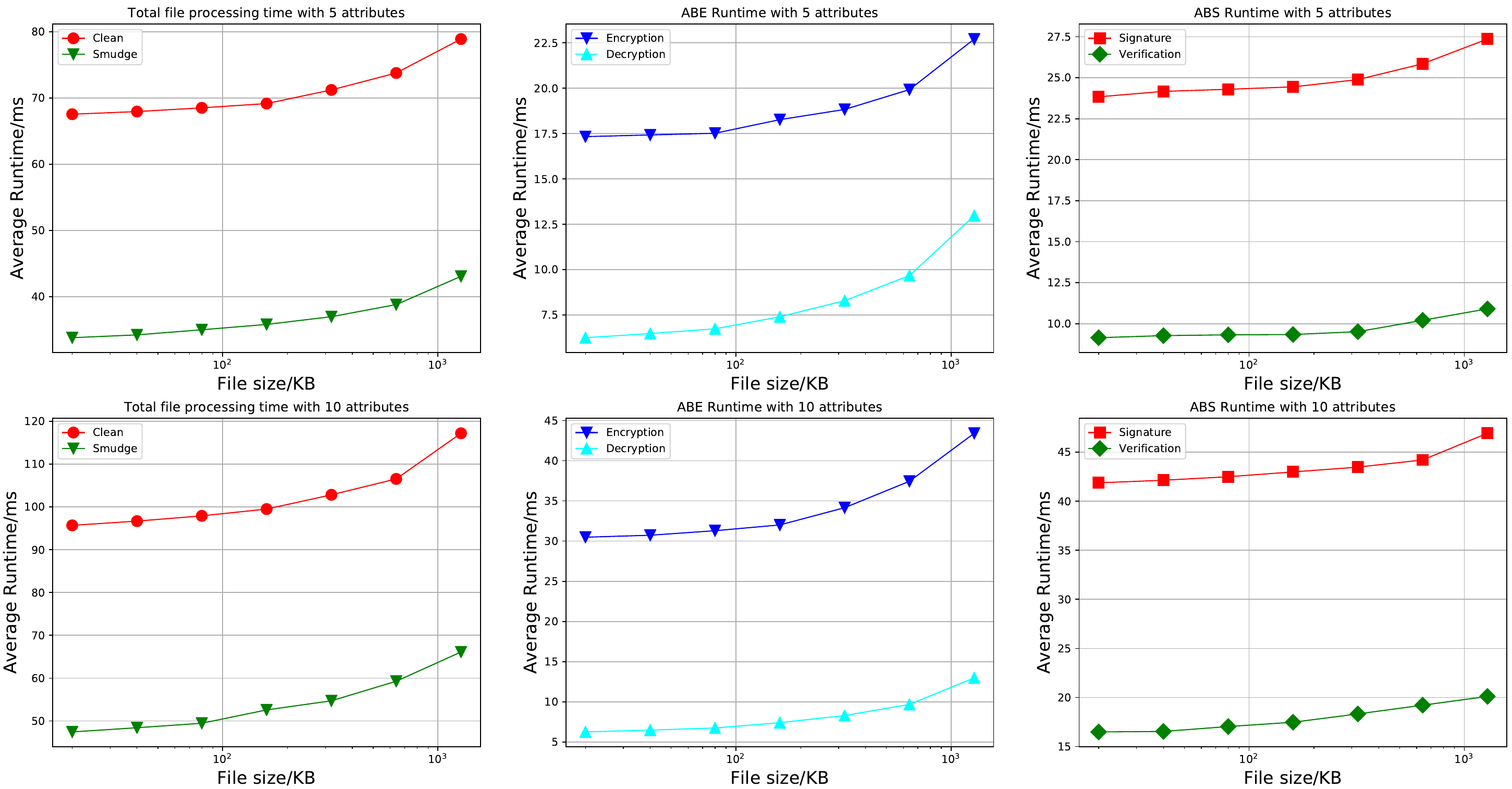}
  \caption{The total file processing time.}\label{fig:totaltime}
\end{figure*}

\noindent\textbf{Setting.}
We have implemented DVAC and integrated it with the Git client.
DVAC relies on ABE and ABS to enforce the access control, and integrates with Git through Git filters, i.e., using a clean module for ABE encryption and ABS signing before sharing data, and a smudge module to verify (ABS) signature and (ABE) decrypt the obtained data. The Git client integrated with DVAC was deployed in a  MacBook Pro laptop with 4 $\times$ 2.3 GHz Intel Core i7-1068NG7 CPUs and 16 GB 3733 MHz LPDDR4X RAM.
The metadata were stored in the Ethereum testnet Ropsten \cite{ropsten}.
We adopt the ABE algorithm \cite{AgrawalC17} and ABS algorithm \cite{maji2011attribute} in the prototype implementation,
 and implement them using C language based on  the PBC (pairing-based cryptography) library \cite{pbclibrary}.

To evaluate the performance overhead in the practical environment,
 we measured the introduced latency and storage for different access control policies (one with 5 attributes and the other with 10) and file sizes.
The range of file size was set as 0--1 MB, according to our statistics on GitHub dataset with Google BigQuery \cite{bigquery}, which is based on the GitHub Archive Project \cite{archive}. We analyzed the complete snapshots of the content of more than 2.8 million open-source GitHub repositories until 19 August 2023.
As shown in Table~\ref{tab:filesize},
   77.44\% files are smaller than 10KB, and 99.42\% files are smaller than 1MB.

\begin{table*}[pos=t,width=0.8\textwidth]
\centering
\caption{Time delta overhead for committing and checking files.}\label{tab:compare1}
\resizebox{0.8\linewidth}{!}{
\begin{tabular}{c|ccc|ccc}
\hline
\multirow{2}{*}{} & \multicolumn{3}{c|}{\textbf{Amortized commit time difference(s)}} & \multicolumn{3}{c}{\textbf{Amortized checkout time difference(s)}} \\ \cline{2-7} 
                  & DVAC(git)               & gringotts             & git-crypt             & DVAC(git)              & gringotts             & git-crypt             \\ \hline
bootstrap         & 0.256               & 0.16                  & 0.125               & 0.128              & 0.042                 & 0.106                 \\
d3                & 0.238               & 0.168                & 0.188                 & 0.155              & 0.056                 & 0.151                 \\
electron          & 0.217               & 0.167                 & 0.134                 & 0.143              & 0.056                 & 0.094                 \\
flutter           & 0.249               & 0.161                 & 0.179                 & 0.152              & 0.091                 & 0.145                 \\
linux             & 0.252               & 0.17                & 0.337                 & 0.198              & 0.207                 & 0.224                 \\
ohmyzsh           & 0.221               & 0.167                 & 0.097                 & 0.113              & 0.042                 & 0.017                 \\
react             & 0.234               & 0.184                 & 0.13                  & 0.125              & 0.068                 & 0.059                 \\
tensorflow        & 0.371               & 0.243                 & 1.061                 & 0.271              & 0.216                 & 0.421                 \\
vscode            & 0.218               & 0.14                  & 0.174                 & 0.184              & 0.135                 & 0.141                 \\
vue               & 0.249               & 0.159                 & 0.125                 & 0.112               & 0.036                 & 0.091                 \\ \hline
\textbf{Average}  & 0.2505              & 0.1719                & 0.255                 & 0.1581             & 0.0949                & 0.1449                \\ \hline
\end{tabular}
}
\end{table*}

\begin{table}[pos=t]
\centering
\caption{File size statistics on Github (until 2023.8.19).}
\label{tab:filesize}

\begin{tabular*}{\linewidth}{@{\extracolsep{\fill}}lrr}
\toprule
File size & Number & Percentage \\
\midrule
0--10KB   & 205563927 & 77.44\% \\
10KB--1MB & 57544225  & 21.98\% \\
1M--10M   & 2097378   & 0.79\%  \\
10M+      & 236104    & 0.09\%  \\
\bottomrule
\end{tabular*}

\end{table}

\subsubsection{Comparison}
We compared the performance of DVAC with existing Git access control schemes Gringotts \cite{xu2023gringotts} and git-crypt \cite{gitcrypt}.

We  selected the top 10 repositories with the most stars on GitHub for the experiment. To measure the commit delta time, we replayed 5000 commits in the new repository. DVAC encrypts and signs during \verb+git add+, so the extra time includes measuring the time to \verb+git add+ all changes in a commit. For Gringotts and git-crypt, this includes the extra time spent encrypting and decrypting files in a commit. Regarding checkout time, we \verb+git checkout+ the first 5000 commits in the clone's repository and measure the extra time. All repositories were configured with five access policies, each containing 10 attributes in DVAC and Gringotts, and a coarse-grained access policy for 10 users in git-crypt.

From Table~\ref{tab:compare1}, we can see that the time overhead used by DVAC is close to but slightly lower than git-crypt, and the overhead is higher than Gringotts. This is because Gringotts uses ECDSA for branch-level write access control, while DVAC employs ABS for finer file-level write access control, offering a better user experience for DVCS. Specifically, \verb+git add+ requires ABS signature, and \verb+git checkout+ requires an ABS verification, thus bringing some overhead, which is also explained in detail below. 
As shown in Section~\ref{subsec:comparison}, unlike Gringotts, DVAC operates without the need for a central server, enabling fully distributed access control. Furthermore, the file-level read and write access control limits permissions to each file version, ensuring that user collaboration, including merge operations, is not disrupted. This also facilitates the implementation of subsequent permission revocations and other policies, allowing DVAC to be used as a plug-in for various DVCS. According to the functional comparison in Table~\ref{tab:compare}, DVAC maintains a time cost within the millisecond range, providing the aforementioned security features without affecting the overall user experience.
Overall, DVAC minimizes changes to existing DVCS (such as Git) while implementing distributed mandatory read and write access control independent of a central server, and provides a more adaptable and secure access control scheme for DVCS at an acceptable cost.

\subsubsection{Latency}
We measured the latency for \verb+git add+ passing through the clean module and \verb+git checkout+ passing through the smudge module for different file sizes and access control policies.
As shown in Figure~\ref{fig:totaltime}, for an access control policy with five attributes, the processing time is less than 78 ms and 43 ms for \verb+git add+ and \verb+checkout+ with a 1 MB file, respectively, and the processing time reduced to 67 ms and 33 ms for a 10 KB file.
For an access control policy with 10 attributes, the processing time increases slightly, to 94 ms and 46 ms for \verb+git add+ and \verb+checkout+ with a 10 KB file, respectively.

In the measurement, the default access policy is not modified. It is assumed that the access policy is cached locally after the metadata is obtained from the blockchain for the first time.

The introduced latency consists of three parts: ABE and ABS processing, metadata querying and file transmission between the Git client and the Git filter driver.

\noindent{\textbf{ABE \& ABS processing.}} ABE and ABS processing requires more than 50\% of the whole processing time. As shown in Figure~\ref{fig:totaltime}, we measure the time overhead of ABE and ABS in DVAC for files with different sizes and access policies with 5 or 10 attributes. The ABE part contains AES-128 symmetric encryption. 
ABE and ABS processing time increases with file size, and the ABE processing time increases more obviously because the file content needs to be encrypted. ABS only hashes the contents of the file and is relatively unaffected.

ABE encryption and ABS processing time also increase with the number of attributes of the policy and user,  and ABE decryption time is not affected by attributes and policies. For example, for a used 10 KB file, the time for ABE encryption and ABS signature increases from 40 ms to 70 ms when the number of attributes in the policy increases from 5 to 10.
The \verb+git add+ operation needs more time than \verb+git pull+, as ABE encryption and ABS signing  require more exponentiations than ABE decryption and ABS verification.

\noindent{\textbf{Metadata querying.}}
DVAC synchronizes blocks of the Ethereum testnet and queries the metadata locally, which requires about 80 ms when the number of blocks is about 500,000.
However, the synchronization of blocks introduces a significant storage overhead; for example, 75 GB is required for 500,000 blocks.
To eliminate the storage requirement on the client, DVAC supports to query the metadata through a third-party service (e.g., Infura \cite{infura}), which needs 100 ms in our network environment. In addition, the efficiency of metadata querying can also be further improved through methods such as \cite{WuPGYX22}.

The transfer latency of metadata creation and modification is similar to the latency of metadata acquisition, both of which are relatively infrequently used and performed by the data owner  through web interfaces  without integration into the Git client.

\noindent{\textbf{File transmission between Git client and git filter driver.}}
In DVAC, the transmission of files between the Git client and the Git filter driver occurs using the long-running process protocol.
The latency for file transmission is contingent upon the file size, ranging from less than 7 ms for a 1 MB file to 1500 ms for a 10 MB file, and escalating to 98,000 ms for a 100 MB file (the maximum file size in GitHub).
 Nevertheless, given that sensitive data is typically stored in a separate file and $99.42\%$ of files are less than 1 MB,
 the latency for file transmission remains within acceptable bounds.

\subsubsection{Storage}
DVAC introduces two extra storage requirements for the Git client.
Firstly, each client needs to store the blocks synchronized from the Ethereum testnet, which amounts to 75 GB when the number of blocks reaches approximately 500,000. This storage overhead can be eliminated by using a third-party service (e.g., Infura \cite{infura}) for metadata querying.
Secondly, each file is stored in the form of ABE ciphertext and ABS signature, rather than the original data.
In our setting, this necessitates approximately 5.16 KB of extra storage for each file when a 10-attribute access control policy is employed.

\subsubsection{ETH Overhead}
In the metadata management part, DVAC need to interact with smart contracts, which inevitably involves the overhead of publishing transactions. In the experimental environment, an upload transaction of metadata requires about ETH 0.0005, while a query of metadata does not cost Ether. In addition, if DVCS is built on the basis of multiple alliances or organizations, metadata management can be carried out on the alliance chain, which can greatly reduce the transaction overhead while improving security. Therefore, the overhead of Ether is very small.

%% file: tex/8_discussion.tex
\section{Discussion}
DVAC controls read and write permissions using ABE and ABS, respectively, while adapting to DVCS characteristics. It is not tied to any specific DVCS but addresses universal version control and metadata management needs. We developed version control and metadata protocols tailored to Git, enhancing DVAC's compatibility with it. However, this approach is also applicable to other DVCS systems, which all require version control and distributed security management. Adapting DVAC to different DVCS involves adjusting based on each system's characteristics, primarily affecting the configuration environment and metadata management deployment, as discussed in Section \ref{sec:disacingit}.

Regarding attribute-based signcryption algorithms(ABSC), most existing ABSC algorithms support only a single policy mode, such as ciphertext-policy or key-policy.Hybrid schemes \cite{yu2017attribute, yu2020lh} combine KP-ABS with CP-ABE but still use a single access policy during the signcryption phase, and generate signature keys based on signature access policies, leading to the need for users to maintain multiple signature keys, especially with varying write permissions across files. In DVAC, we bind signature and decryption keys to attributes separately and bind signatures and ciphertexts to read and write access policies per file. This minimizes user key maintenance overhead and facilitates file permission control. Thus, although ABSC algorithms offer higher computational efficiency and lower communication costs, due to sacrificing some flexibility, existing ABSC algorithms cannot be directly applied to DVAC. Nevertheless, ABSC in principle supports the extension of different access policies, and its integration concept provides us with future optimization directions, so we will consider the possibility of combining ABE with ABS  further improve computing efficiency in the future.

%% file: tex/9_conclusion.tex
\section{Conclusion}
\label{sec:conclusion}
In this paper, we propose DVAC, an enforcing cryptographic access control scheme for fully distributed VCS, aimed at preventing unauthorized read or write access to sensitive data in distributed repositories. According to the operation requirements and version control characteristics of DVCS, DVAC sets the enforcing cryptographic access control protocol and version control adaptation protocol.
DVAC implements access control on distributed nodes, protecting data access without relying on a central server, and DVAC's architecture integrates seamlessly with existing DVCS, necessitating only a few straightforward configurations to ensure compatibility. The participation of smart contracts guarantees that only the data owner retains complete control over the shared data and the consistency and security of multi-node access control. We have implemented DVAC, integrated it with Git, and the performance evaluation demonstrates that the introduced overhead is acceptable in practical environment.

Our current scheme also has some limitations, such as the high computational cost of attribute-based cryptography and the undiscussed fine-grained permission dynamic management. Future research will focus on optimizing distributed system access control, such as multi-authority fine-grained retractions, and developing smart contract-assisted computing schemes for resource-limited devices, in order to improve system efficiency and the scope of application.

%% file: tex/10_appendix.tex
\section*{Acknowledgements}
This work was supported by National Natural Science Foundation of China (No.82090053). The authors wish to thank anonymous reviewers for their valuable comments and suggestions that improved this paper.